%% file: du-arxiv-draft.tex
\nonstopmode
\documentclass[a4,11pt,en]{quick-document}

\usepackage{tikz, pgf}
\usepackage{tikz-qtree}

\def\flip{\mbox{\textnormal{flip()}}}

\newenvironment{algowithlines}[1][0.6\linewidth]{\bgroup\minipage{#1}\vspace{1em}\algorithmic[1]}{\endalgorithmic\vspace{0.8em}\endminipage\egroup}
\newenvironment{algo}[1][0.6\linewidth]{\bgroup\minipage{#1}\vspace{1em}\algorithmic}{\endalgorithmic\vspace{0.8em}\endminipage\egroup}
\def\imagesPath#1{images/#1}
\newcommand{\includeImage}[2][scale=1.00]{%
    \includegraphics[#1]{\imagesPath{#2}}}

\title{Optimal Discrete Uniform Generation from Coin~Flips, and Applications}

\author{J\'{e}r\'{e}mie Lumbroso}
\date{\today}

\def\isFDRarticle{1}

\begin{document}

\maketitle

\begin{abstract}
  This article introduces an algorithm to draw random discrete uniform
  variables within a given range of size $n$ from a source of random
  bits.  The algorithm aims to be simple to implement and optimal both
  with regards to the amount of random bits consumed, and from a
  computational perspective---allowing for faster and more efficient
  Monte-Carlo simulations in computational physics and biology. I also
  provide a detailed analysis of the number of bits that are spent per
  variate, and offer some extensions and applications, in particular
  to the optimal random generation of permutations.
\end{abstract}

\input{du-v3b-body.tex}

\bibliographystyle{plain}
\bibliography{discrete-uniform}

\appendix

\input{du-v3-implementation}

\input{du-appendix-ratbern}

\end{document}

%% file: du-v3b-body.tex
\def\algFDR{\textsc{Fast Dice Roller}}
\def\algFDRs{\textsc{FDR}}
\newcommand{\Mellin}[2][x,s]{\ensuremath{\mathcal{M}\!\left[#2,#1\right]}}
\newcommand{\MellinStrip}[1][\alpha,\beta]{\ensuremath{<\!\!{#1}\!\!>}}

\noindent Now that simulations can be run extremely fast, they are
routinely able to consume over a billion random variates an hour---and
several orders of magnitude more throughout an execution. At such a point,
the amount of pseudo-randomness at our disposal may eventually become
a real issue, and it is pertinent to devise techniques that are economical
with respect to the amount of randomness consumed, while remaining as or
more efficient than existing techniques with regards to speed and space
usage.

\paragraph{The random-bit model.}

Much research has gone into simulating probability distributions, with
most algorithms designed using infinitely precise \emph{continuous uniform
  random} variables (see \cite[II.3.7]{Devroye86}). But because
(pseudo-)randomness on computers is typically provided as 32-bit
integers---and even bypassing issues of true randomness and bias---this
model is questionable. Indeed as these integers have a fixed precision,
two questions arise: when are they not precise enough? when are they too
precise? These are questions which are usually ignored in typical
fixed-precision implementations of the aforementioned algorithms. And it
suggests the usefulness of a model where the unit of randomness is not the
uniform random variable, but the \emph{random bit}.

This random bit model was first suggested by Von Neumann~\cite{Neumann51},
who humorously objected to the use of fixed-precision pseudo-random
uniform variates in conjunction with transcendant functions approximated
by truncated series\footnote{Or in his words, as related by Forsythe:
  ``\textit{I have a feeling, however, that it is somehow silly to take
    a random number and put it elaborately into a power series.}''}. His
remarks and algorithms spurred a fruitful line of theoretical research
seeking to determine \emph{which} probabilities can be simulated using
only random bits (unbiased or biased? with known or unknown bias?), with
which complexity (expected number of bits used?), and which guarantees
(finite or infinite algorithms? exponential or heavy-tailed time
distribution?). Within the context of this article, we will focus on
designing practical algorithms using unbiased random bits.

In 1976, Knuth and Yao~\cite{KnYa76} provided a rigorous theoretical
framework, which described generic optimal algorithms able to simulate any
distribution. These algorithms were generally not practically usable:
their description was made as an infinite tree---infinite not only in the
sense that the algorithm terminates with probability $1$ (an unavoidable
fact for any probability that does not have a finite binary expansion),
but also in the sense that the description of the tree is infinite and
requires an infinite precision arithmetic to calculate the binary
expansion of the probabilities.

In 1997, Han and Hoshi~\cite{HaHo97} provided the \emph{interval
  algorithm}, which can be seen as both a generalization and
implementation of Knuth and Yao's model. Using a random bit stream, this
algorithm amounts to simulating a probability $p$ by doing a binary search
in the unit interval: splitting the main interval into two equal
subintervals and recurse into the subinterval which contains $p$. This
approach naturally extends to splitting the interval in more than two
subintervals, not necessarily equal. Unlike Knuth and Yao's model, the
interval algorithm is a concrete algorithm which can be readily
programmed... as long as you have access to arbitrary precision arithmetic
(since the interval can be split to arbitrarily small sizes).

In 2003, Uyematsu and Li~\cite{UyLi03} gave implementations of the
interval algorithm which use a fixed precision integer arithmetic, but
these algorithms approximate the distributions (with an accuracy that is
exponentially better as the size of the words with which the algorithm
gets to work is increased) even in simple cases, such as the discrete
uniform distribution which they use as an illustrating example using
$n=3$.

I was introduced to this problematic through the work of Flajolet,
Pelletier and Soria~\cite{FlPeSo11} on \emph{Buffon machines}, which are
a framework of probabilistic algorithms allowing to simulate a wide range
of probabilities using only a source of random bits.

\paragraph{Discrete uniform distribution.}

Beyond these generic approaches, there has been much interest specifically
in the design of efficient algorithms to sample from the discrete uniform
distribution. While it is the building brick to many other more
complicated algorithms, it is also notable for being extremely common in
various types of simulations in physics, chemistry, etc.

Wu~\etal~\cite{WuHuOu02} were among the first to concretely consider the
goal of saving bits as much as possibly: they note that in practice small
range uniform variables are often used, and thus slice a 32-bit
pseudo-random integer into many smaller integers which they then reject.
Although they do this using a complicated scheme of Boolean functions, the
advantages are presumably that the operations being on a purely bitwise
level, they may be done by hardware implementations and in parallel.

Orlov~\cite{Orlov09} gives an algorithm which reduces the amount of
rejection per call, and assembles a nice algorithm using efficient
bit-level tricks to avoid costly divisions/modulo as much as possible; yet
his algorithm still consumes randomness as 32-bit integers, and is
wasteful for small values (which are typically the most common in
simulations).

Finally Ladd~\cite{Ladd09} considers the problem of drawing from specific
(concentrated) distributions for statistical physics simulations; he does
this by first defining a lookup table, and then uniformly drawing indexes
in this table. His paper is notable to us because it is written with the
intent of making as efficient a use of random bits as possible, and
because he provides concrete implementations of his algorithms. However,
as he draw discrete uniform variables simply by truncating 32-bit
integers, his issue remains the same: unless his lookup table has a size
which is a power of two, he must contend with costly rejection which
increases running time, in his simulations, more than fourfold (see his
compared results for a distribution with 8 states, and with 6 states).

\paragraph{Our contribution.}

Our main algorithm allows for the exact sampling of discrete uniform
variables using an optimal number of random bits for any range $n$.

It is an extremely efficient implementation of Knuth and Yao's general
framework for the special case of the discrete uniform distribution:
conceptually simple, requiring only $2\log n$ bits of storage to draw
a random discrete uniform variable of range $n$, it is also practically
efficient to the extent that it generally improves or matches previous
approaches. A full implementation in C/C++ is provided as illustration at
the end of the article, in \ifdefined\isFDRarticle %
Appendix~\ref{sec:fdr-imp}.\else %
Section~\ref{sec:fdr-imp}.\fi %

Using the Mellin transform we precisely quantify the expected number of
bits that are used, and exhibit the small fluctuations inherent in the
base conversion problem. As expected, the average number of bits used is
slightly less good than the information-theoretic optimality---drawing
a discrete uniform variable comes with a small toll---and so we show how
using a simple (known) encoding scheme we can quickly reach this
information-theoretic optimality. Finally, using a similar method, we
provide likewise optimal sampling of random permutations.


\section{The {\algFDR} algorithm\label{sec:alg}}

The {\algFDR} algorithm, hereafter abbreviated {\algFDRs}, is very simple,
and can be easily implemented in a variety of languages (taking care to
use the \emph{shifting} operation to implement multiplication by $2$). It
takes as input a fixed integer value of $n$, and returns as output
a uniformly chosen integer from $\set{0, \ldots, n-1}$. The $\flip$
instruction does an unbiased coin flip, that is it returns $0$ or $1$ with
equiprobability. Both this instruction (as a buffer for a PRNG which
generates 32-bit integers) and the full algorithm are implemented in C/C++
at the end of the article, in \ifdefined\isFDRarticle %
Appendix~\ref{sec:fdr-imp}.\else %
Section~\ref{sec:fdr-imp}.\fi %

\begin{theorem}
  The {\algFDR} algorithm described below returns an integer which is
  uniformly drawn from the set $\set{0, \ldots, n-1}$ and terminates with
  probability $1$.
\end{theorem}

\begin{center}\begin{algowithlines}
  \Function{FastDiceRoller}{$n$}
  \State$v\gets 1;\;\;c\gets 0$
  \Loop
    \State $v \gets 2v$
    \State $c \gets 2c + \flip$
    \If{$v \geqslant n$}
      \If{$c<n$}
        \State\Return $c$
      \Else
        \State$v \gets v - n$
        \State$c \gets c - n$
      \EndIf
    \EndIf
  \EndLoop
  \EndFunction
\end{algowithlines}\end{center}

\begin{proof}
  Consider this statement, which is a loop invariant: $c$ is uniformly
  distributed over $\set{0,\ldots,v-1}$. Indeed, it is trivially true at
  initialization, and:
  \begin{itemize}
  \item in lines 4 and 5, the range $v$ is doubled; but $c$ is doubled as
    well and added a parity bit to be uniform within the enlarged new range;
  \item lines 10 and 11 are reached only if $c\geqslant n$; conditioned on
    this, $c$ is thus uniform within $\set{n, \ldots, v-1}$, with this
    range containing at least one integer since we also have $v >
    c \geqslant n$; as such we are simply shifting this range when
    substracting $n$ from both $c$ and $v$.
  \end{itemize}
  The correctness of the algorithm follows from this loop invariant. As
  $c$ is always uniformly distributed, when the algorithm returns $c$ upon
  verifying that $v\geqslant n$ (the current range of $c$ is at least $n$)
  and $c < n$ (the actual value of $c$ is within the range we are
  interested in), it returns a uniform integers uniform in $\set{0,\ldots,
    n-1}$.

  The termination and exponential tails can be proved by showing that an
  equivalent, less efficient algorithm is geometrically distributed: in
  this equivalent algorithm, instead of taking care to recycle random bits
  when the condition on line 7 fails, we simply restart the algorithm; by
  doing so, we have an algorithm that has probability
  $p=n/\cramped{2^{\lfloor\cramped{\log_2 n}\rfloor + 1}}> 1/2$ of
  termination every $\lfloor\cramped{\log_2 n}\rfloor + 1$ iterations.
\end{proof}

The space complexity is straightforward: by construction $c < v$ and $v <
2n$ are always true; thus $c$ and $v$ each require $1+\cramped{\log_2 n}$
bits. The time complexity, which also happens to be the random bit
complexity, since exactly one random bit is used per iteration, is more
complicated and detailed in the following section.

\begin{remark*}
  Most random generation packages do not come with a flip or random
  boolean operation (and those which do provide such commodity usually do
  so in a grossly inefficient way). Thus a concrete way of consuming
  random bits is to hold 32-bit random numbers in a temporary buffer
  variable and use each bit one after the other. What is surprising is
  that the overhead this introduces is more than compensated by the
  savings it brings in random bits---which are costly to generate.
\end{remark*}

\begin{remark*}
  It should be noted that this algorithm can be straightforwardly
  adapted to the problem of simulating a Bernoulli law of rational
  parameter $p$, as illustrated in Appendix~\ref{sec:fdr-ratbern}.
\end{remark*}

\pagebreak[1]
\section{Analysis of the expected cost in random bits}

\begin{figure}
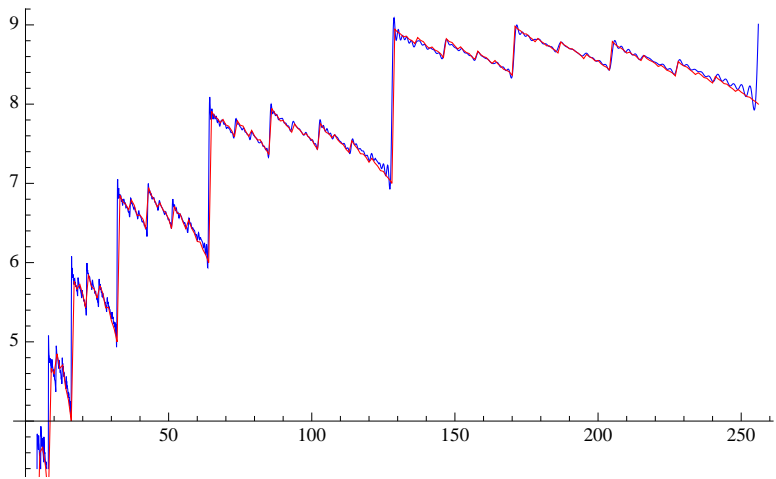

  \centering
  \includeImage{du-touched-up-mellin1.pdf}
  \caption{\label{fig:unifcost-mellin}Plot of the expected cost of
    generating a random discrete uniform variable, where the $x$-axis is
    the (discrete) range $n$ of the variable: the red curve is computed
    from the exact sum, the blue curve is computed from the asymptotic
    expression of Theorem~\ref{thm:unifcost} (using only a dozen roots in
    the trigonometric polynomial $P$).}
\end{figure}

\begin{theorem}\label{thm:unifcost}
  The expected number $u_n$ of random bits needed to randomly draw
  a \emph{uniform} integer from a range of length $n$ using the {\algFDRs}
  algorithm is asymptotically
  \begin{align*}
    u_n = \log_2 n + \frac 1 2 + \frac 1 {\log 2} - \frac
    {\gamma}{\log 2} + P(\cramped{\log_2} n) + O(\cramped{n^{-\alpha}})
  \end{align*}
  for any $\alpha > 0$, and where $P$ is a periodic function,
  a trigonometrical polynomial defined in Equation~\eqref{eq:poly-p}.
\end{theorem}

\begin{remark*}
  Furthermore, the {\algFDRs} algorithm terminates with probability 1, and
  as proven in the previous section, the distribution of its running time
  (and number of random bits) has \emph{exponential tails}\footnote{A
    random variable $X$ is said to have exponential tails if there are
    constants $C$ and $\rho < 1$ such that $\prob{X > k} \leqslant
    C \rho^k$.}.
\end{remark*}

\noindent The remainder of this section is dedicated to proving this
theorem. First, we will revisit some of Knuth and Yao's main results, and
by showing the {\algFDRs} algorithm is an implementation of their
theoretical framework to the special case of the discrete uniform
distribution, we obtain an expression of its expected cost in random bits
as an infinite sum. Then, we use a classical tool from analysis of
algorithms, the Mellin transform, to obtain the sharp asymptotic estimate
stated in the theorem.

\subsection{\label{subsec:ky-ddg}Knuth and Yao's optimal DDG-trees}

As mentioned, Knuth and Yao introduce a class of algorithms, called
DDG-trees, to optimally simulate discrete distributions using only random
bits as a randomness source.

The discrete distributions are defined by a probability vector
$\mathbd{p}=(\cramped{p_1,\ldots,p_n})$, which is possibly infinite.
A DDG-tree is a tree (also possibly infinite) where: internal nodes
indicate coin flips; external nodes indicate outcomes; at depth $k$ there
is an external node labeled with outcome $i$ if and only if the $k$-th bit
in the binary expansion of $\cramped{p_i}$ is $1$.

Knuth and Yao do not provide an efficient way to build, or simulate,
these DDG-trees---and this is far from a trivial matter. But one of
their main results~\cite[Theorem 2.1]{KnYa76} is that DDG-trees
provide simulations which are optimal in number of random bits used,
with an average complexity which, if finite, is
\begin{align}
  \nu(\mathbd{p}) = \nu(\cramped{p_1}) + \ldots + \nu(\cramped{p_n})
\end{align}
where $\nu$ is a function defined by
\begin{align}\label{eq:ky-nu}
  \nu(x) = \sum_{k=0}^\infty \frac{\left\{2^k\,x\right\}}{2^k}\text{.}
\end{align}
where $x\mapsto \{x\}$ denotes the fractional part function.
A straightforward consequence is that the optimal average random-bit
complexity to simulate the discrete uniform distribution, that is where
$\mathbd{p} = (1/n,\ldots,1/n)$, is
\begin{align}\label{eq:ky-u-sum}
  u_n = n\sum_{k=0}^\infty \left\{\frac{2^k}{n}\right\}\frac{1}{2^k}\text{.}
\end{align}

\subsection{The {\algFDRs} algorithm, as an implementation of a DDG-tree}

This sum is the exact complexity of the algorithm presented in
Section~\ref{sec:alg}; this is best understood by noticing that the
{\algFDRs} algorithm is an implementation of a DDG-tree. The algorithm
efficiently computes the binary expansion of $1/n$: every iteration
computes a single bit, and those iteration where the condition line~6 is
verified are those where this bit is equal to $1$, and where, according to
Knuth and Yao's framework, the DDG-tree should have $n$ terminal leaves.
The variable $c$ simulates a path within the tree, from the root to one of
the leaves.

\begin{figure}[h]
  \centering
  \begin{minipage}{6cm}
    \def\tn#1{{\bf #1}}
    \begin{tikzpicture}
      \Tree [.0 
              [.0 
                [.0
                  [.{\tn 0} ]
                  [.{\tn 1} ] ]
                [.1 
                  [.{\tn 2} ]
                  [.{\tn 3} ] ] ]
              [.1
                [.2 
                  [.{\tn 4} ]
                  [.0 
                    [.{\tn 0} ]
                    [.{\tn 1} ] ] ]
                [.3 
                  [.1
                    [.{\tn 2} ]
                    [.{\tn 3} ] ]
                  [.2
                    [.{\tn 4} ]
                    [.0 \edge[roof]; {... } ] ] ] ] ]
    \end{tikzpicture}
  \end{minipage}
  \caption{\label{fig:fdr-as-ddg}Tree of the branching process of the
    {\algFDRs} algorithm for the case where $n=5$; a single process is
    a path from the root to a leaf, where each internal node corresponds
    to a new random bit being used, and the leaves correspond to the
    outcome. This tree is exactly a DDG-tree for simulating the uniform
    distribution of rang $n=5$. Indeed, note that the binary expansion of
    $1/5$ is periodic and goes $1/5=0.001100110011\ldots_2$, which
    corresponds with the alternance of levels with and without leaves in
    the corresponding tree.}
\end{figure}
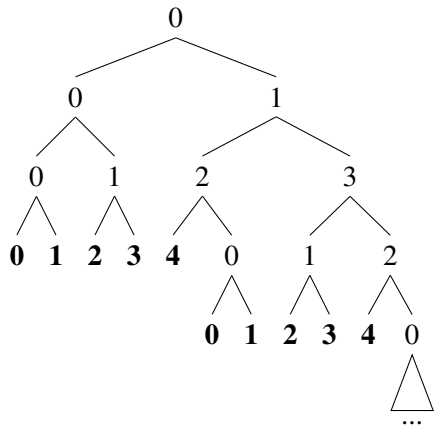

\subsection{The Mellin transform of a base function\label{sub:mellin1}}

The Mellin transform is a technique to obtain the asymptotics of some
types of oscillating functions. It is of central importance to the
analysis of algorithms because such oscillating functions appear naturally
(for instance, in most all analyses having to do with divide-and-conquer
type recursions for instance), and their periodic behavior can generally
not be quantified with other usual, and less precise asymptotic techniques

\begin{definition}
  Let $f$ be a locally Lebesgue-integrable function over $(0,+\infty)$.
  The \emph{Mellin transform} of $f$ is defined by as the complex
  function, $s\in\mathbb{C}$,
  \begin{align*}
    \Mellin{f(x)} = f^\star(s) = \int_0^{+\infty} f(x) x^{s-1}\drm
    x\text{.}
  \end{align*}
  The largest open strip $\alpha < \mathrm{Re}(s) < \beta$ in which the
  integral converges is called the \emph{fundamental strip}. We may note
  this strip \MellinStrip. 
\end{definition}

\paragraph{Important properties.}

Let $f$ be a real function; the following $F$ is called a \emph{harmonic
  sum} as it represents the linear superposition of ``harmonics'' of the
base function $f$,
\begin{align}
  F(x) = \sum_k \lambda_k f(\cramped{\mu_k} x)
\end{align}
and the $\cramped{\lambda_k}$ are called the \emph{amplitudes}, the
$\cramped{\mu_k}$ the \emph{frequencies}~\cite{FlGoDu95}. Most functions like
$F$ usually involve subtle \emph{fluctuations} which preclude the use of
real asymptotic techniques.

The first important property we will make use of is that the Mellin
transform allows us to \emph{separate} the behavior of the base function
from that of its harmonics. Indeed, if $f^\star(s)$ is the Mellin
transform of $f$, the Mellin transform of the harmonic sum $F$ involving
$f$ is then simply
\begin{align}\label{eq:mellin-harm}
  F^\star(s) = \left(\sum_k\frac {\lambda_k}{{\mu_k}^s}\right) \cdot
    f^\star(s)\text{.}
\end{align}
The second property that is central to the analyses in this chapter is
that the behavior of a function $f$ in $0$ and in $+\infty$ can be
directly respectively read on the poles to the left or right of the Mellin
transform $\cramped{f^\star}$.

\subsection{Mellin transform of the fractional part}

Before proceeding to the Mellin transform of the harmonic sum which we are
interested in, using the principles described by
Flajolet~\etal~\cite{FlGoDu95}, we must manually calculate the Mellin
transform of the fractional part function using classical integration
tools---thus giving another proof of a result which seems to be due to
Titchmarsh~{\cite[\S 2]{Titchmarsh86}}.

\begin{lemma}\label{lemma:mellin-frac}
  Let $f(x)=\left\{1/x\right\}$ be the fractional part of the inverse
  function, its Mellin transform, valid for $0 < \RealPart{s} < 1$, is
  \begin{align*}
    f^\star(s) = -\frac {\zeta(s)} s\text{.}
  \end{align*}
\end{lemma}

\begin{proof}
  From the observation that, for $x>1$, $f(x) = 1/x$, we may split the
  integral,
  \begin{align*}
     \int_0^\infty \left\{\frac 1 x\right\} x^{s-1} \drm x =
     \int_0^1 \left\{\frac 1 x\right\} x^{s-1} \drm x +
     \int_1^\infty x^{s-2} \drm x\text{.}
  \end{align*}
  To integrate on the unit interval, we split according to the inverses of
  integers,
  \begin{align*}
    \sum_{n=1}^{\infty} \int_{\frac 1 {n+1}}^{\frac{1}{n}} \left\{\frac
      1 x\right\} x^{s-1} \drm x &=
    \int_0^1 x^{s-2}\drm x -
    \sum_{n=1}^\infty n \int_{\frac 1 {n+1}}^{\frac{1}{n}} x^{s-1} \drm x
  \end{align*}
  and furthermore
  \begin{align*}
    -\sum_{n=1}^\infty n \int_{\frac 1 {n+1}}^{\frac{1}{n}} x^{s-1} \drm
    x = -\frac 1 s \sum_{n=1}^\infty \frac 1 {n^{s-1}} +\frac
    1 s \sum_{n=1}^\infty \frac 1 {(n+1)^{s-1}} -\frac
    1 s \sum_{n=1}^\infty \frac 1 {(n+1)^{s}}\text{.}
  \end{align*}
  Each sum can be replaced by properly shifted Riemann's zeta function,
  \begin{align*}
    -\sum_{n=1}^\infty n \int_{\frac 1 {n+1}}^{\frac{1}{n}} x^{s-1} \drm x =
    -\frac{\zeta(s)}{s}\text{.}
  \end{align*}
  By analytic continuation, the two integrals of $\cramped{x^{s-2}}$ are
  valid even outside of their initial domain of definition. They cancel
  each other out, and we are left with
  \begin{align*}
    f^\star(s) := \int_0^\infty \left\{\frac 1 x\right\} x^{s-1} \drm x =
    -\frac{\zeta(s)}{s}\text{.}
  \end{align*}
  Finally, the fundamental strip in which this Mellin transform is defined
  can be found by observing that
  \begin{align*}
    \lim\limits_{x\to 0} \left\{\frac 1 x\right\} = O(1) = O(\cramped{x^0}) %
    \quad\text{and}\quad %
    \forall x > 1, \left\{\frac 1 x\right\} = \frac 1 x = O(\cramped{x^{-1}})\text{.}
  \end{align*}
\end{proof}

\begin{remark*}
  Observe that we calculate the Mellin transform of $\left\{1/x\right\}$,
  but this also provides the Mellin transform of $\left\{x\right\}$.
  Indeed this follows the following functional property
  \begin{align*}
    \Mellin{f(x)} = f^\star(s) %
    \qquad\Leftrightarrow\qquad %
    \Mellin{f(1/x)} = -f^\star(-s)
  \end{align*}
  respectively on the fundamental strips $\MellinStrip[\alpha,\beta]$ and
  $\MellinStrip[-\beta, -\alpha]$, which is a special case of a more
  general rule expressing the Mellin transform $f(x^\eta)$, see for
  instance~\cite[Theorem~1]{FlGoDu95}.
\end{remark*}

\subsection{Mellin transform of the discrete uniform average complexity}

We now have all the tools to study the harmonic sum we are interested,
\begin{align}
  u_n = n\sum_{k=0}^\infty \left\{\frac{2^k}{n}\right\}\frac{1}{2^k}\text{.}
\end{align}
Our first step is to transform this into a real function (replace the
discrete variable~$n$ by a real variable~$x$) and decompose this as a base
function and harmonics,
\begin{align}
  F(x) := x \sum_{k=0}^\infty \left\{\frac{2^k}{x}\right\}\frac{1}{2^k} =
  x \sum_{k=0} f(\cramped{2^{-k}}\, x)\, 2^{-k}\quad\text{with}\quad
  f(x) = \left\{\frac{1}{x}\right\}\text{.}
\end{align}
We now use: the functional property we have recalled in
Equation~\eqref{eq:mellin-harm}; the additional
property~\cite[Fig.~1]{FlGoDu95} the Mellin transform of $xf(x)$ is
\begin{align}
  \Mellin{xf(x)}=f^\star(x+1)
\end{align}
on the shifted fundamental strip $\MellinStrip[\alpha-1, \beta-1]$; and
the Mellin transform of the fractional part, as stated in
Lemma~\ref{lemma:mellin-frac}. With these, we finally obtain that
\begin{align}
  F^\star(s) = - \frac{\zeta(s+1)}{(1-\cramped{2^s})(s+1)}\text{.}
\end{align}
This Mellin transform is defined on the fundamental strip $-1
<\mathrm{Re}(s)<0$, and it has one double pole in $s=0$ (from the
$\zeta(s+1)$ and the denominator $1-\cramped{2^s}$ which cancels out)
which will induce a logarithmic factor, and an infinity of simple complex
poles in $s=2\cplxI k \pi/\log 2$ from which will come the fluctuations.

Indeed, using the Mellin summation formula~\cite[p.~27]{FlGoDu95}, we
obtain the asymptotic expansion
\begin{align}
  F(x) \msim -\sum_{s\in \Omega} \mathrm{Res}(\cramped{F^\star(s)x^{-s}})
\end{align}
where $\Omega$ is the set of poles to the right of the fundamental strip,
which we have just enumerated. Thus we get
\begin{align*}
  F(x) \msim \log_2 x + \frac 1 2 + \frac 1 {\log 2} - \frac{\gamma}{\log 2}  +
  P(\cramped{\log_2} x) + O(\cramped{n^{-\alpha}})\text{,}
\end{align*}
for any arbitrary $\alpha > 0$, and with $P$ a trigonometric polynomial
defined as:
\begin{align}\label{eq:poly-p}
  P(\cramped{\log_2} x) := - \frac 1 {\log 2}
  \sum_{k\in\Z\setminus\set{0}} \frac {\zeta(2\cplxI k \pi/\log 2 +
    1)}{2\cplxI k \pi/\log 2 + 1} \lgexp{-2\cplxI k \pi \log_2 x}\text{.}
\end{align}


\section{Algorithmic tricks to attain entropic optimality}

In the expression of the average random-bit cost, we can distinguish two
parts,
\begin{align*}
  \log_2 n \qquad\text{and}\qquad %
  t_n = \frac 1 2 + \frac 1 {\log 2} - \frac{\gamma}{\log 2} +
  P(\cramped{\log_2} n) + O(\cramped{n^{-\alpha}})\text{.}
\end{align*}
On one side, the expected $\cramped{\log_2} n$ contribution that comes
from ``encoding'' $n$ in a binary base (using random bits); on the other,
some additional \emph{toll} when $n$ is not a dyadic rational, i.e. $n
\not= \cramped{2^k}$, and the generation process requires rejection.

\begin{figure}[h]
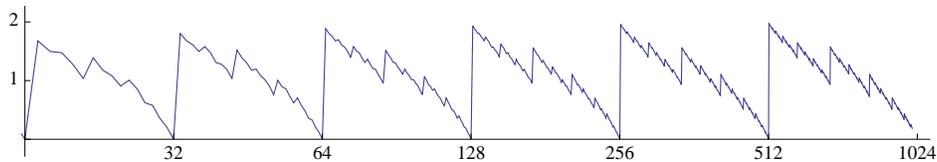

  \centering
  \includeImage{du-toll-plot.pdf}
  \caption{\label{fig:du-toll}Plot of the toll $\cramped{t_n}$ in the
    average random-bit cost of generating a discrete uniform variable,
    with $n$ from 2 to 1024; as expected the plot exhibits a distinct
    logarithmic period. When $1/n$ has a finite binary expansion, i.e. it
    is a \emph{dyadic rational} with $n=\cramped{2^{-k}}$, the toll is
    equal to zero.}
\end{figure}

\noindent In their article, Knuth and Yao~\cite[Theorem~2.2]{KnYa76} prove
that for all discrete distributions (including the one we are interested
in), this toll has the following bounds:
\begin{align}\label{eq:ky-toll-bound}
  0 \leqslant t_n \leqslant 2\text{.}
\end{align}

Because this toll, as exhibited in Figure~\ref{fig:du-toll}, is not
monotonous and upperbounded by a constant, the implication is that it is
generally more efficient (in terms of the proportion of bits which are
wasted in the toll) to generate a discrete uniform of larger range,
because this toll becomes of insignificant magnitude compared to the main
logarithmic factor.

In the remainder of this section, we use this observation to go beyond
theoretic bounds and reach the entropic optimality for the generation of
discrete uniform variables, and random permutations.

\subsection{Batch generation}

As it has been observed many times, for instance by Han and
Hoshi~\cite[V.]{HaHo97}, it can prove more efficient to generate the
Cartesian product of several uniform variables, than generating a single
uniform variable---especially when the considered range is small.

Thus, instead of generating a single discrete uniform variable of
range~$n$, we generate $j$ variables at a time by drawing a discrete
uniform variable $Y$ of range~$n^j$ and we use its decomposition in
$n$-ary base
\begin{align}
  Y := X_{j} \cdot n^{j-1} + \ldots + X_1 \cdot n^0
\end{align}
to recover the $\cramped{X_i}$ through a simple (albeit slightly costly)
succession of integer divisions. As it turns out, this trick decreases the
toll by a more than linear factor, as encapsulated by the following
theorem.

\begin{theorem}\label{thm:manyunifcost}
  The number $u_{n,j}$ of random bits needed to randomly draw a uniform
  integer from a range of length $n$ increases when the random integers
  are drawn $j$ at a time,
  \begin{align*}
    u_{n,j} = \log_2 n + \frac 1 j \left(\frac 1 2 + \frac 1 {\log 2} - \frac
    {\gamma}{\log 2}\right) + \frac 1 {j^2} P(\cramped{\log_2} n) +
    O(\cramped{n^{-\alpha}}/j)
  \end{align*}
  for some any $\alpha > 0$, so that as $j$ tends to infinity, we reach
  asymptotical information theoretic optimality
  \begin{align*}
    u_{n,\infty} \msim \log_2 n\text{.}
  \end{align*}
\end{theorem}

\noindent In practice, because of the quadratic rate by which this trick
attenuates the importance of the oscillations, it doesn't take much to get
very close to the information theoretic optimum of $\cramped{\log_2 n}$
bits. As illustrated by Figure~\ref{fig:unifcost-many}, typically taking
$j=6$ is already very good, and a significant improvement over $j=1$.

\begin{figure}
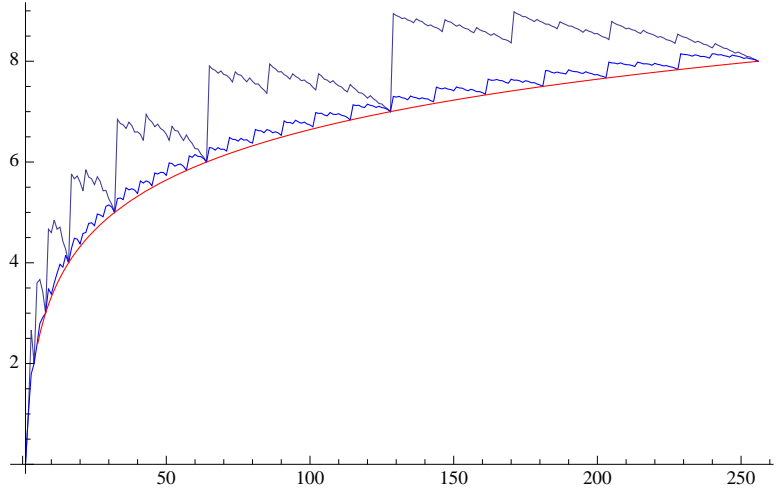

  \centering
  \includeImage{du-compare-many.pdf}
  \caption{\label{fig:unifcost-many}This plot illustrates the quadratic
    decrease of the periodic oscillations, and fast convergence of the
    logarithm, when generating several random discrete uniform variables
    at a time. In dark blue, the expected cost of generating one uniform
    variable; in light blue, the expected cost when generating six at
    a time (i.e., $j=6$); as a point of comparison, in red, the
    information-theoretic optimal given by the binary logarithm.}
\end{figure}

Larger values of $j$ should be disregarded as, considering the size of
words is finite (32, 64, or 128 bits, typically, depending on the computer
architecture), it is of course important to keep in mind that $j$ must be
chosen so as to not cause an overflow.

\begin{proof}
  Using the function $F(x)$ as defined in the proof of
  Theorem~\ref{thm:unifcost}, we define $F(x) = G(\cramped{x^j})/j$.
  Classical rules of the Mellin transform show that
  \begin{align*}
    G^\star(s) = \frac{F^\star(s/j)}{j^2}
  \end{align*}
  valid in the fundamental strip where $-j < \RealPart{s} < 0$. We now can
  define $\cramped{t_{n,j}}$ the unitary toll for each discrete uniform
  variable of range $n$, when they are generated $j$ at a time, as
  \begin{align*}
    t_{n,j} = \frac 1 j \left(\frac 1 2 + \frac 1 {\log 2} - \frac
    {\gamma}{\log 2}\right) + \frac 1 {j^2} P(\cramped{\log_2} n) +
    O(\cramped{n^{-\alpha}}/j)\text{.}
  \end{align*}
  Using Knuth and Yao's bound restated in
  Equation~\eqref{eq:ky-toll-bound}, we have the rough bound
  \begin{align*}
    0 \leqslant t_{n,j} \leqslant \frac 2 j
  \end{align*}
  which is sufficient to show that as $j$ tends to infinity,
  $\cramped{t_{n,j}}$ tends to zero.
\end{proof}

\subsection{Optimally generating random permutations}

Beside (both continuous and discrete) uniform variables, another
elementary building block in random generation is \emph{random
  permutations}. They are routinely used in a great many deal of
applications\footnote{And even so, their generation may still prove
  challenging, as recently evidenced by Microsoft. Indeed, as a concession
  to the European Union (which found Microsoft guilty of imposing its own
  browser to Windows users, to the detriment of the competition),
  Microsoft provided Windows~7 users with a randomly permutated ballot
  screen for users to select a browser. But a programming error made the
  ordering far from ``random'' (uniform), which briefly caused
  a scandal.}.

For example, on a related topic, one such application is to the automatic
combinatorial random generation methods formalized by Flajolet and his
collaborators. In both the recursive method~\cite[\S 3]{FlZiVa94} and
Boltzmann sampling~\cite[\S 4]{DuFlLoSc04}, only the shapes of labeled
objects are sampled: the labeling can then be added after, by drawing
a random permutation of the size of the object.

But random permutations are also useful in other fields. In statistical
physics~\cite[\S 1.2.2]{Krauth06}, for instance in quantum physic
simulations.

It should be noted that the algorithmic ideas presented in this subsection
are classical~\cite[\S XIII]{Devroye86}. Their relevance in the context of
this article is that, in conjunction with the {\algFDRs} algorithm, they
allow to concretely attain previously theoretical-only optimal lower
bounds---while for the most part remaining reasonably efficient.

\paragraph{Asymptotical optimality using the Fisher-Yates shuffle}

\begin{figure}
  \centering
  \begin{algo}[19.8em]
    \Function{Shuffle}{$T$}
    \State $n \gets |T|$
    \For{$i$}{$1$}{$n$}
      \State $k \gets i+\Call{DiscreteUniform}{n-i+1}$
      \State \Call{Swap}{$T$, $i$, $k$}
    \EndFor
    \EndFunction
  \end{algo}
  \caption{The Fisher-Yates random shuffle as described by Durstenfeld.
    The array $T$ is indexed in one, and \textsc{DiscreteUniform}
    concordingly returns a random uniform number in the range $0$ to $n-i+1$
    included.}
\end{figure}

A straightforward idea with the elements thus far presented, is as
follows. Assuming one generates the uniform in batches with $j$ sufficient
so that we may assume that each uniform of range $N$ takes $\log_2 N +
\varepsilon$ bits, with some very small $\varepsilon$, then the random bit
complexity $C$ of generating a permutation with the Fisher-Yates shuffle
is

\begin{align}
  C \msim \log_2 2 + \log_2 3 + \ldots_2 + \log_2 n + \varepsilon n =
  \log_2 n! + \varepsilon n
\end{align}

But unfortunately, even though by generating the separate uniform
variables in large enough batches, we can decrease the toll considerably
for each uniform variable, we will still have an overall linear amount of
such toll when considering the $n$ variables that must be drawn.
Furthermore it is not very practical to have to generate uniform in
batches (this supposes that we are drawing many permutations at the same
time). So we suggest another solution.

\paragraph{Optimal efficiency using a succinct encoding.}

As early on as 1888, Laisant~\cite{Laisant1888} exhibited a more direct
way of generating random permutations, using a mixed radix decomposition
called \emph{factorial base system}.

\begin{lemma}
  Let $U$ be a uniformly drawn integer from $\set{0,\ldots,n!-1}$, and let
  $\cramped{X_n}$ be the sequence such that
  \begin{align*}
    U = X_{n} \cdot (n-1)! + \ldots + X_1\cdot 0!\qquad\text{and}\qquad
    \forall i, 0 \leqslant X_i < i
  \end{align*}
  then the $\cramped{X_i}$ are \emph{independent} uniformly drawn integers
  from $\set{0,\ldots, i-1}$.
\end{lemma}

Laisant observed that a permutation could then be constructed by taking
the sorted list of elements, and taking the $\cramped{X_n}$-th as first
element of the permutation, then the $\cramped{X_{n-1}}$-th of the
remaining elements as second element of the permutation, and so
on\footnote{This idea is often associated with Lehmer who rediscovered
  it~\cite{Lehmer60}.}. What is remarkable is that using this
construction, it is possible to directly compute the number of inversions
of the resulting permutation by summing all $\cramped{X_i}$, where
inversions $I(\sigma)$ of a permutation $\sigma$ are defined as
\begin{align}
  I(\sigma) := \set{ (i,j) \in \N^2\ |\ i < j \text{ and } \sigma_i >
    \sigma_j}\text{.}
\end{align}
Unfortunately the algorithm requires the use of a chained list instead of
an array, and thus has quadratic time complexity---which is
prohibitive\footnote{Nevertheless it is notable that summing the
  $\cramped{X_i}$ (without needing to compute the actual permutation)
  yields an interesting way of generating a random variable distributed as
  the number of inversions in a permutation of size $n$.}.

We can use a different bijection of integers with permutations which can
be computed in linear time, by simply using the $\cramped{X_i}$ as input
for the previously described Fisher-Yates shuffle. In this way, we
optimally generate a random permutation from a discrete uniform variate of
range $n!-1$, and show how to attain information theoretic optimality in
a much less contrived way than described in the previous subsection.

A caveat though is that word size may become a real issue: with 32-bit
registers one can obtain permutations up to $n=12$; with 64-bit, up to
$n=21$; with 128-bit up to $n=33$.

\begin{remark*}
  The general idea of numbering or indexing (i.e., establishing
  a bijection with a continuous range of integers containing zero) all
  objects of a combinatorial class of a given size is often called
  \emph{ranking} and the inverse transformation---obtaining an object from
  its rank---is called \emph{unranking}, but has also been referred to as
  the \emph{decoding method}~\cite[\S XIII.1.2]{Devroye86}.

  For a long time, devising such unranking schemes often relied on luck or
  combinatorial acumen. Mart\'{\i}nez and Molinero~\cite{MaMo01}
  eventually established a general approach by adapting the previously
  mentioned recursive random generation method of
  Flajolet~\etal~\cite{FlZiVa94}. While this approach is not necessarily
  efficient, it provides a usable algorithm to attain random-bit
  optimality for the random generation of many combinatorial structures.
\end{remark*}

\section{Conclusion}

It would have been conceivable that this article yield a theoretical
algorithm of which the sole virtue would have been to provide
a theoretical optimal complexity, while proving less than useful for
practical use.

But unexpectedly, it turns out that the extra buffering inherent in
consuming randomness random-bit-by-random-bit\footnote{The implementation
  of the $\mathrm{flip}()$ function. It involves: drawing a random 32-bit
  int, storing it in a temporary variable, and then extracting each bit as
  needed, while making sure to refill the variable once all 32 bits have
  been used.}, although time consuming, is more than compensated by the
increased efficiency in using random bits compared with most common
methods.

It remains to be seen whether this is still the case on newer CPUs which
contain embedded instructions for hardware pseudo random generation.
However there are arguments that support this: first, assuming that
hardware pseudo random generation is to eventually become widespread
enough for software to take advantage of it, it seems likely to take
a significant time to be common; second, the computer architecture shift
seems to be towards RISC architectures which are not burdened with such
complex instructions.

\paragraph{Prospective future work.}

The result presented here interestingly yields, as a direct consequence,
the expected cost of the \emph{alias method}, a popular method to simulate
discrete distributions which are known explicitly as a histogram, also
known as \emph{sampling with replacement}. This method is often said to
have constant time complexity, but that is under the model where discrete
uniform variables are generated in constant time.
 
There are many different applications which are still to be examined:
several classical algorithms, which use discrete (and continuous) uniform
random variables, where the random bit cost is as of yet unknown.

Of particular interest, the process known as \emph{sampling without
  replacement}, or sampling from a discrete distribution which evolves in
a fairly predictable manner. The most promising algorithms for this
problem follow the work of Wong and Easton~\cite{WoEa80}, which uses
a partial sum tree. It remains to be seen what is the overall bit
complexity of this algorithm, and whether it can be improved (for instance
by choosing a specific type of tree).

\section*{Acknowledgments}

I am very grateful to Mich\`{e}le Soria for her careful reading of
drafts of this article, and her valuable input; I would also like to
thank Philippe Dumas for discussions on the Mellin analysis and Axel
Bacher for a discussion on Lehmer codes. Finally, I wish to warmly
thank Kirone Mallick for his encouragement in my pursuit of concrete
applications for this algorithm, in theoretical and statistical
physics.

%% file: du-v3-implementation.tex
\section{\label{sec:fdr-imp}Implementation of the main {\algFDRs} algorithm}

This proposed implementation makes use of two non-standard packages: the
Boost library's standard integer definition, to make sure that the buffer
integer variable has the correct size; and the Mersenne Twister algorithm
for the random generation of the 32-bit integers themselves (the code for
this is wildly available). Note that using the correct integer type on
a given machine will do; as will using another random integer generator
than MT algorithm.

\begin{verbatim}
 #include <cstdlib>
 #include <boost/cstdint.hpp>    // Fixed size integers
 #include "mt19937ar.hpp"        // Mersenne Twister
 
 using namespace std;
 
 // For benchmarking purposes only
 static uint64_t flip_count = 0;
 
 // Flip buffering variables
 static uint32_t flip_word = 0;
 static int flip_pos = 0;
 
 int flip(void)
 {
   if(flip_pos == 0) {
     flip_word = genrand_int32();
     flip_pos = 32;
   }
 
   flip_count++;
   flip_pos--;
   return (flip_word & (1 << flip_pos)) >> flip_pos;
 }

 inline uint32_t algFDR(unsigned int n)
 {
   uint32_t v = 1, c = 0;
   while(true)
   {
     v = v << 1;
     c = (c << 1) + flip();
     if(v >= n)
     {
       if(c < n) return c;
       else
       {
         v = v - n;
         c = c - n;
       }
     }
   }
 }
\end{verbatim}

%% file: du-appendix-ratbern.tex
\section{\label{sec:fdr-ratbern}Simulating rational Bernoulli variables}

There is a well-known idea in random
generation~\cite[XV.1.2]{Devroye86b}, to efficiently draw a random
Bernoulli variable of parameter $p$: draw a geometric random variable
of parameter $1/2$, $k\in\GeometricLaw{1/2}$; then return, as result
of the Bernoulli trial, the $k$-th bit in the dyadic representation of
$p$.

Interestingly, this idea was already known to physicists, as evidenced
by an early paper by Pierre~\etal~\cite{PiGiSc87}, but seems not to be
commonly used today in Monte-Carlo implementations. Internal
simulations show that for typical Boltzmann energy simulations drawing
Bernoulli variables in this way consumes 16 times fewer random bits,
and that simulations are accelerated by a 4 to 6 factor (this is less
impressive than the number of saved bits because of the accounting
overhead required to buffer 32-bit integers into single flips).

The limitation of this algorithm is that obtaining the dyadic
representation of any $p$ is not a trivial matter. Fortunately for
rational numbers it is simple enough, and although this is not a new
contribution, for the sake of completeness we illustrate it in
Figure~\ref{fig:ratbern}.

\begin{figure}[h]
  \centering
  \raisebox{-\height}{{\begin{algo}[13.5em]
    \Function{BinaryBase}{$k/n$}
    \State $v \gets k$
    \Loop
      \State $v \gets 2v$
      \If{$v\geqslant n$}
        \State $v\gets v - n$
        \State {output} 1
      \Else
        \State {output} 0
      \EndIf
    \EndLoop
    \EndFunction
  \end{algo}}}\hspace{2em}\raisebox{-\height}{{\begin{algo}[12.5em]
    \Function{Bernoulli}{$k/n$}
    \State $v \gets k$
    \Repeat
      \State $v \gets 2v$
      \If{$v\geqslant n$}
        \State $v\gets v - n$
        \State $b\gets 1$
      \Else
        \State $b\gets 0$
      \EndIf
    \Until{flip() = 1}
    \State\Return $b$
    \EndFunction
  \end{algo}}}%
{
  \caption{\label{fig:ratbern}A simple algorithm to output the binary
    decomposition of a rational $k/n$, $k<n$, and the corresponding
    algorithm that simulates a Bernoulli distribution of parameter
    $p=k/n$. Neither algorithm is dependent on the required precision of
    the binary expansion. They both use only $1+\log_2 n$ bits of space,
    and require only one shift, one substraction and one comparison per
    iteration. The Bernoulli simulation algorithm consumes on average two
    flips (random bits).}}
\end{figure}

\begin{remark*}
  With Knuth and Yao's theorem, this algorithm can be shown to be
  optimal: indeed, it simply requires drawing a geometric variable of
  parameter $1/2$, which takes on average $2$ bits. Coincidentally,
  that is the optimal cost of drawing any Bernoulli variable. Recall
  the $\nu$ function defined as,
  \begin{align*}\tag{\ref{eq:ky-nu}}
    \nu(x) = \sum_{k=0}^\infty \frac{\left\{2^k\,x\right\}}{2^k}\text{.}
  \end{align*}
  From the results recalled in Subsection~\ref{subsec:ky-ddg}, we have
  that the optimal average cost of drawing a random Bernoulli variable
  of parameter $p$ is,
  \begin{align*}
    \nu(p)+\nu(1-p) &= \sum_{k=0}^\infty \frac{\left\{2^k\,p\right\}}{2^k} +
    \sum_{k=0}^\infty \frac{\left\{2^k\,(1-p)\right\}}{2^k}\\
    &= \sum_{k=0}^\infty \frac{\left\{2^k\,p\right\}+
      \left\{2^k\,(1-p)\right\}}{2^k} %
    = \sum_{k=0}^\infty \frac{1}{2^k} = 2
  \end{align*}
  hence the optimality.
\end{remark*}